\documentclass[10pt,amssymb,amsfont,a4paper]{article}
\usepackage{latexsym,amssymb,amsmath,amsthm,color}
\usepackage{geometry}
\usepackage{url}
\usepackage{graphicx}
\usepackage[utf8]{inputenc}
\usepackage[numbers]{natbib}
\usepackage{authblk}
\usepackage{tikz}
\usetikzlibrary{matrix}
\usetikzlibrary{shapes}
\usetikzlibrary{decorations.pathreplacing,matrix}
\usepackage{arydshln}

\theoremstyle{definition}
\newtheorem{definition}{Definition}
\newtheorem{example}[definition]{Example}
\newtheorem{construction}{Construction}

\theoremstyle{plain}
\newtheorem{theorem}{Theorem}
\newtheorem{proposition}[definition]{Proposition}
\newtheorem{lemma}[definition]{Lemma}
\newtheorem{remark}[definition]{Remark}

\title{New constructions of MSRD codes}
\author{Umberto Mart{\'i}nez-Pe\~{n}as \thanks{umberto.martinez@uva.es}}
\affil{IMUVa-Mathematics Research Institute,\\University of Valladolid, Spain}

\date{}

\begin{document}

\maketitle

\begin{abstract}
In this work, we provide four methods for constructing new maximum sum-rank distance (MSRD) codes. The first method, a variant of cartesian products, allows faster decoding than known MSRD codes of the same parameters. The other three methods allow us to extend or modify existing MSRD codes in order to obtain new explicit MSRD codes for sets of matrix sizes (numbers of rows and columns in different blocks) that were not attainable by previous constructions. In this way, we show that MSRD codes exist (by giving explicit constructions) for new ranges of parameters, in particular with different numbers of rows and columns at different positions.


\end{abstract}

\section{Introduction} \label{sec intro}

The sum-rank metric, defined in \cite{multishot} and implicitly considered earlier in \cite{spacetime-kumar}, has recently attracted considerable attention in Coding Theory due to its applications in reliable and secure multishot network coding \cite{secure-multishot, multishot}, PMDS codes for repair in distributed storage \cite{cai-MR, gopi-MR, universal-lrc}, rate-diversity optimal space-time codes \cite{spacetime-kumar, Mohannad-Journal}, and multilayer crisscross error correction \cite{multicover}, among others.

The size or dimension (when linear) of codes also satisfy a Singleton bound with respect to their minimum sum-rank distance \cite[Th. III.2]{alberto-fundamental}. Codes attining this bound are therefore optimal with respect to the size-distance tradeoff and are called maximum sum-rank distance (MSRD) codes. Linearized Reed--Solomon codes \cite{linearizedRS} are the first MSRD codes that can be decoded in polynomial time over a field of subexponential size in the code length \cite{secure-multishot}. Afterwards, a number of alternative MSRD codes have appeared in the literature \cite{alberto-fundamental, chen, linearizedRS, generalMSRD, doubly, twisted, neri-oneweight, santonastaso}, covering other ranges of parameters (different field sizes and/or matrix sizes).

In this work, we provide four methods for constructing new MSRD codes. The first method (Section \ref{sec cart products}) consists of a special arrangement of cartesian products of preexisting MSRD codes and allows faster decoding than known MSRD codes of the same parameters. The other three methods (Sections \ref{sec combining}, \ref{sec lattices} and \ref{sec systematic}) allow us to extend or modify existing MSRD codes in order to obtain new explicit MSRD codes for sets of matrix sizes (numbers of rows and columns in different blocks) that were not attainable by previous constructions. Furthermore, the constructions in Sections \ref{sec lattices} and \ref{sec systematic} admit different numbers of rows and columns at different positions. Not many explicit MSRD constructions with this feature were known before \cite{alberto-fundamental, chen}. In Section \ref{sec comparisons}, we compare the concrete examples of MSRD codes obtained in this work with the known MSRD codes from the literature. In particular, we show that the parameters of MSRD codes from the literature can all be attained by our constructions, whereas our constructions of MSRD codes attain new ranges of parameters (numbers of rows and columns).

\section{Preliminaries} \label{sec preliminaries}

In this preliminary section, we revisit the basic properties of codes in the sum-rank metric (Subsection \ref{subsec preliminaries sum-rank}) and some known constructions of MSRD codes (Subsection \ref{subsec preliminaries msrd}). For tutorials and surveys on the topic, we refer to \cite{sum-rank-chapter, fnt}. 

Let $ \mathbb{F}_q $ denote the finite field of size $ q $, denote by $ \mathbb{F}_q^{m \times n} $ the space of matrices of size $ m \times n $ over $ \mathbb{F}_q $, for positive integers $ m $ and $ n $, and set $ \mathbb{F}_q^n = \mathbb{F}_q^{1 \times n} $. We also denote $ \mathbb{N} = \{ 0,1,2,\ldots \} $, $ [n] = \{ 1,2, \ldots, n \} $ and $ [m,n] = \{ m, m+1, \ldots, n \} $ for positive integers $ m $ and $ n $ with $ m \leq n $. In the following, $ \langle \cdot \rangle_{\mathbb{F}_q} $ and $ \dim_{\mathbb{F}_q} $  denote linear span and dimension over $ \mathbb{F}_q $.

\subsection{The sum-rank metric} \label{subsec preliminaries sum-rank}

Fix positive integers $ \ell $, $ m_1 \geq m_2 \geq \ldots \geq m_\ell $ and $ n_i \leq m_i $, for $ i \in [\ell] $. We will consider the sum-rank metric in the space $ \prod_{i=1}^\ell \mathbb{F}_q^{m_i \times n_i} $, where we will call each factor $ \mathbb{F}_q^{m_i \times n_i} $ a rank block, thus $ \ell $ is the number of (rank) blocks. For $ C = (C_1,\ldots, C_\ell) \in \prod_{i=1}^\ell \mathbb{F}_q^{m_i \times n_i} $, we define its sum-rank weight as
$$ {\rm wt}(C) = \sum_{i=1}^\ell {\rm Rk}(C_i), $$
where $ {\rm Rk} $ denotes the rank function. The sum-rank metric is defined as $ {\rm d} (C,D) = {\rm wt}(C-D) $, for $ C, D \in \prod_{i=1}^\ell \mathbb{F}_q^{m_i \times n_i} $. For a code (i.e., a subset) $ \mathcal{C} \subseteq \prod_{i=1}^\ell \mathbb{F}_q^{m_i \times n_i} $, we define its minimum sum-rank distance as
$$ {\rm d}(\mathcal{C}) = \min \{ {\rm d}(C,D) : C,D \in \mathcal{C}, C \neq D \}. $$
For an $ \mathbb{F}_q $-linear code $ \mathcal{C} \subseteq \prod_{i=1}^\ell \mathbb{F}_q^{m_i \times n_i} $, its minimum sum-rank distance coincides with its minimum sum-rank weight, that is, $ {\rm d}(\mathcal{C}) = \min \{ {\rm wt}(C) : C \in \mathcal{C}, C \neq 0 \} $. 

Observe that, when $ \ell = 1 $, the sum-rank metric recovers the rank metric, and when $ m_1 = n_1 = \ldots = m_\ell = n_\ell = 1 $, the sum-rank metric recovers the Hamming metric.

As in the case of the Hamming metric, there exists a Singleton bound that relates the minimum sum-rank distance and the size of a code without involving the field size (except for taking logarithms or dimensions). For a code (linear or non-linear) $ \mathcal{C} \subseteq \prod_{i=1}^\ell \mathbb{F}_q^{m_i \times n_i} $ with $ |\mathcal{C}| \geq 2 $, let $ {\rm d}(\mathcal{C}) = \sum_{i=1}^{j-1} n_i + \delta + 1 $, where $ j \in [\ell] $ and $ 0 \leq \delta \leq n_j - 1 $. The Singleton bound for the sum-rank metric, proven in \cite[Th. III.2]{alberto-fundamental}, reads
\begin{equation}
\log_q |\mathcal{C}| \leq \sum_{i=j}^\ell m_in_i - m_j \delta.
\label{eq singleton bound}
\end{equation}
Notice that, if $ \mathcal{C} $ is $ \mathbb{F}_q $-linear, then $ \log_q |\mathcal{C}| = \dim_{\mathbb{F}_q}(\mathcal{C}) $. A code $ \mathcal{C} \subseteq \prod_{i=1}^\ell \mathbb{F}_q^{m_i \times n_i} $ is called a Maximum Sum-Rank Distance (MSRD) code if it meets the Singleton bound (\ref{eq singleton bound}). See Subsection \ref{subsec preliminaries msrd} for some known explicit constructions.

When $ m = m_1 = \ldots = m_\ell $, we may consider the space $ \mathbb{F}_{q^m}^n $, where $ n = n_1 + \cdots + n_\ell $, instead of $ \prod_{i=1}^\ell \mathbb{F}_q^{m_i \times n_i} $, due to the following. Given an ordered basis $ \boldsymbol\gamma = (\gamma_1, \ldots, \gamma_m) \in \mathbb{F}_{q^m}^m $ of $ \mathbb{F}_{q^m} $ over $ \mathbb{F}_q $, we define the $ \mathbb{F}_q $-linear vector space isomorphism $ M_{\boldsymbol\gamma}^r : \mathbb{F}_{q^m}^r \longrightarrow \mathbb{F}_q^{m \times r} $ given by
\begin{equation}
M_{\boldsymbol\gamma}^r (\mathbf{c}) = \left( \begin{array}{cccc}
c_{1,1} & c_{1,2} & \ldots & c_{1,r} \\
c_{2,1} & c_{2,2} & \ldots & c_{2,r} \\
\vdots & \vdots & \ddots & \vdots \\
c_{m,1} & c_{m,2} & \ldots & c_{m,r} 
\end{array} \right),
\label{eq def matrix representation}
\end{equation}
for $ \mathbf{c} = (c_1, \ldots,c_r ) \in \mathbb{F}_{q^m}^r $, where $ c_{i,j} \in \mathbb{F}_q $, for $ i \in [m] $ and $ j \in [r] $, are the unique scalars such that $ c_j = \sum_{i=1}^m \gamma_i c_{i,j} $, for $ j \in [r] $. Now, if we set $ \mathbf{n} = (n_1,\ldots, n_\ell) $, we may extend the previous map to another $ \mathbb{F}_q $-linear vector space isomorphism $ M_{\boldsymbol\gamma}^{\mathbf{n}} : \mathbb{F}_{q^m}^n \longrightarrow \prod_{i=1}^\ell \mathbb{F}_q^{m \times n_i} $ by 
\begin{equation}
M_{\boldsymbol\gamma}^{\mathbf{n}}(\mathbf{c}) = \left( M_{\boldsymbol\gamma}^{n_1} (\mathbf{c}_1), \ldots, M_{\boldsymbol\gamma}^{n_\ell} (\mathbf{c}_\ell) \right),
\label{eq def matrix repr tuples}
\end{equation}
for a vector $ \mathbf{c} = (\mathbf{c}_1, \ldots, \mathbf{c}_\ell) \in \mathbb{F}_{q^m}^n $, where $ \mathbf{c}_i \in \mathbb{F}_{q^m}^{n_i} $, for $ i \in [\ell] $. We may also define its sum-rank weight as
$$ {\rm wt}(\mathbf{c}) = {\rm wt}\left( M_{\boldsymbol\gamma}^{\mathbf{n}}(\mathbf{c}) \right) = \sum_{i=1}^\ell {\rm Rk}\left( M_{\boldsymbol\gamma}^{n_i} (\mathbf{c}_i) \right). $$
Therefore, we may define the sum-rank metric in $ \mathbb{F}_{q^m}^n $ simply as $ {\rm d}(\mathbf{c},\mathbf{d}) = {\rm wt}(\mathbf{c} - \mathbf{d}) $, for $ \mathbf{c}, \mathbf{d} \in \mathbb{F}_{q^m}^n $. The advantage of considering the sum-rank metric in $ \mathbb{F}_{q^m}^n $ is that we may consider $ \mathbb{F}_{q^m} $-linear codes in such an ambient space. Notice that most constructions of MSRD codes are $ \mathbb{F}_{q^m} $-linear codes in $ \mathbb{F}_{q^m}^n $ \cite{linearizedRS, generalMSRD, doubly, twisted}, see Subsection \ref{subsec preliminaries msrd}. However, in this manuscript we will construct $ \mathbb{F}_q $-linear MSRD codes where not all $ m_1 , \ldots, m_\ell $ are equal. Only a few constructions in this case are known \cite{alberto-fundamental, chen}.  

Observe that, when considering the sum-rank metric in $ \mathbb{F}_{q^m}^n $ as above, we need to specify the vector $ \mathbf{n} = (n_1, \ldots, n_\ell) $, which we call the sum-rank length partition. Otherwise, the map $ M^{\mathbf{n}}_{\boldsymbol\gamma} $ is not well defined.

\subsection{Some known MSRD codes} \label{subsec preliminaries msrd}

We now briefly describe the general $ \mathbb{F}_{q^m} $-linear MSRD codes in $ \mathbb{F}_{q^m}^n $ introduced in \cite{generalMSRD}. They generalize linearized Reed--Solomon codes \cite{linearizedRS}, which were the first $ \mathbb{F}_{q^m} $-linear MSRD codes whose field sizes $ q^m $ are subexponential in the code length $ n $. In general, the MSRD codes in \cite{generalMSRD} are the ones with the smallest finite-field sizes $ q^m $ for the given parameters known so far. Moreover, they have the longest block length $ \ell $ compared to $ q $ and the matrix sizes, among known MSRD codes. Constructions \ref{cons bases}, \ref{cons msrd} and \ref{cons msrd systematic} in this manuscript (Sections \ref{sec combining}, \ref{sec lattices} and \ref{sec systematic}, respectively) will allow us to extend the block length or modify the matrix sizes of such MSRD codes in non-trivial ways. 

Since we are looking for long MSRD codes and an MSRD code can easily be shortened \cite[Sec. 3.3]{gsrws}, we will consider the following codes with the longest lengths possible. Let $ \mu $ and $ r $ be positive integers, define $ \ell = \mu (q-1) $ and $ n = \ell r $, and consider the sum-rank length partition $ \mathbf{n} = (r, \ldots, r) $ ($ \ell $ times). For $ k \in [n] $, define the matrix in $ \mathbb{F}_{q^m}^{k \times n} $ given by
\begin{equation}
 M_k(\mathbf{a},\boldsymbol\beta) = \left( \begin{array}{lll|c|lll}
\beta_1 & \ldots & \beta_{\mu r} & \ldots & \beta_1 & \ldots & \beta_{\mu r} \\
 \beta_1^q a_1 & \ldots & \beta_{\mu r}^q a_1 & \ldots & \beta_1^q a_{q-1} & \ldots & \beta_{\mu r}^q a_{q-1} \\
\beta_1^{q^2} a_1^{\frac{q^2-1}{q-1}} & \ldots & \beta_{\mu r}^{q^2} a_1^{\frac{q^2-1}{q-1}} & \ldots &  \beta_1^{q^2} a_{q-1}^{\frac{q^2-1}{q-1}} & \ldots &  \beta_{\mu r}^{q^2} a_{q-1}^{\frac{q^2-1}{q-1}} \\
\vdots & \ddots & \vdots & \ddots & \vdots & \ddots & \vdots \\
 \beta_1^{q^{k-1}} a_1^{\frac{q^{k-1}-1}{q-1}} & \ldots &  \beta_{\mu r}^{q^{k-1}} a_1^{\frac{q^{k-1}-1}{q-1}} & \ldots &  \beta_1^{q^{k-1}} a_{q-1}^{\frac{q^{k-1}-1}{q-1}} & \ldots &  \beta_{\mu r}^{q^{k-1}} a_{q-1}^{\frac{q^{k-1}-1}{q-1}} \\
\end{array} \right) ,
\label{eq general msrd generator}
\end{equation}
where $ a_1, \ldots, a_{q-1} \in \mathbb{F}_{q^m}^* $ are such that $ N_{q^m,q}(a_i) \neq N_{q^m,q}(a_j) $ if $ i \neq j $ (where $ N_{q^m,q} (a) = a \cdot a^q \cdots a^{q^{m-1}} = a^{\frac{q^m-1}{q-1}} $, for $ a \in \mathbb{F}_{q^m} $, is the norm of $ \mathbb{F}_{q^m} $ over $ \mathbb{F}_q $), and where $ \beta_1, \ldots, \beta_{\mu r} \in \mathbb{F}_{q^m}^* $ are such that, if we set $ \mathcal{H}_i = \left\langle \beta_{(i-1)r+1}, \beta_{(i-1)r+2}, \ldots, \beta_{ir}  \right\rangle_{\mathbb{F}_q} \subseteq \mathbb{F}_{q^m} $, then
\begin{enumerate}
\item
$ \dim_{\mathbb{F}_q}(\mathcal{H}_i) = r $, and
\item
$ \mathcal{H}_i \cap \left( \sum_{j \in \Gamma} \mathcal{H}_j \right) = \{ 0 \} $, for any set $ \Gamma \subseteq [\mu] $, such that $ i \notin \Gamma $ and $ |\Gamma| \leq \min \{ k,\mu \} -1 $,
\end{enumerate}
for all $ i \in [\mu] $.

With these assumptions, the $ \mathbb{F}_{q^m} $-linear code $ \mathcal{C}_k(\mathbf{a},\boldsymbol\beta) = \{ \mathbf{x} M_k(\mathbf{a},\boldsymbol\beta) : \mathbf{x} \in \mathbb{F}_{q^m}^k \} \subseteq \mathbb{F}_{q^m}^n $ has dimension $ k $ (over $ \mathbb{F}_{q^m} $) and is MSRD by \cite[Th. 3.12]{generalMSRD}. We refer the reader to \cite[Sec. 4]{generalMSRD} for concrete examples of choices of $ a_1, \ldots, a_{q-1} $ and $ \beta_1, \ldots, \beta_{\mu r} $ (in particular for the longest values of $ r $ and $ \mu $, and thus of $ \ell $, given $ q $ and $ m $). Recall that, by \cite[Th. 5]{gsrws}, the dual code $ \mathcal{C}_k(\mathbf{a},\boldsymbol\beta)^\perp $ is also MSRD. However, generator matrices of such codes are not known in general.

Linearized Reed--Solomon codes \cite{linearizedRS} correspond to the above MSRD codes when $ \mu = 1 $, that is, $ \boldsymbol\beta = (\beta_1, \ldots, \beta_r) $ and the two conditions on $ \mathcal{H}_1 $ simply mean that $ \beta_1, \ldots, \beta_r $ are $ \mathbb{F}_q $-linearly independent.

\section{Construction 1: Cartesian products} \label{sec cart products}

In general, cartesian products of MSRD codes are not MSRD. However, we now present a particular case where they are indeed MSRD. The main interest in this construction is that, when the component codes are linearized Reed--Solomon codes, we will see that the resulting code admits decoding algorithms that are faster than those of other MSRD codes of the same parameters.

\begin{construction} \label{cons product}
Consider (linear or non-linear) codes $ \mathcal{C}_1, \ldots, \mathcal{C}_t \subseteq \prod_{i=1}^\ell \mathbb{F}_q^{m_i \times m_i} $, where $ m_1 \geq \ldots \geq m_\ell $. Consider their cartesian product arranged as follows:
\begin{equation}
\mathcal{C} = \left\lbrace \left( \left( \begin{array}{c}
C_{1,1} \\
\vdots \\
C_{t,1}
\end{array} \right), \ldots, \left( \begin{array}{c}
C_{1,\ell} \\
\vdots \\
C_{t,\ell}
\end{array} \right) \right) : (C_{k,1},\ldots, C_{k,\ell}) \in \mathcal{C}_k, k \in [t] \right\rbrace \subseteq \prod_{i=1}^\ell \mathbb{F}_q^{(t m_i) \times m_i}, 
\label{eq def cartesian prod}
\end{equation}
and consider the sum-rank metric in $ \prod_{i=1}^\ell \mathbb{F}_q^{(t m_i) \times m_i} $ by taking ranks in each block of matrices $ \mathbb{F}_q^{(t m_i) \times m_i} $, for $ i \in [\ell] $. Observe that this is different than simply considering $ \left( \prod_{i=1}^\ell \mathbb{F}_q^{m_i \times m_i} \right)^t $ with the rank blocks $ \mathbb{F}_q^{m_i \times m_i} $.
\end{construction}

As in the classical case, we have the following basic result. The proof is straightforward.

\begin{lemma}
If $ d_k = {\rm d}(\mathcal{C}_k) $, for $ k \in [t] $, then
$$ \log_q|\mathcal{C}| = \sum_{k=1}^t \log_q|\mathcal{C}_k| \quad \textrm{and} \quad {\rm d}(\mathcal{C}) = \min \{ d_1, \ldots, d_t \}. $$
\end{lemma}

In particular, we obtain MSRD codes in the following particular case.

\begin{theorem}
If $ \mathcal{C}_i $ is MSRD for $ i \in [t] $, $ |\mathcal{C}_1| = \ldots = |\mathcal{C}_t| $ and $ d = d_1 = \ldots = d_t $, then $ \mathcal{C} $ is MSRD. More precisely, $ {\rm d}(\mathcal{C}) = d = \sum_{i=1}^{j-1}m_i + \delta + 1 $, where $ j \in [\ell] $ and $ 0 \leq \delta \leq m_j - 1 $, and
$$ \log_q|\mathcal{C}| = t \left( \sum_{i=j}^\ell m_i^2 - m_j \delta \right). $$
\end{theorem}
\begin{proof}
Since $ \mathcal{C}_k $ is MSRD of distance $ d $, we have $ \log_q|\mathcal{C}_k| = \sum_{i=j}^\ell m_i^2 - m_j \delta $, for $ k \in [t] $, thus $ \log_q|\mathcal{C}| = t \left( \sum_{i=j}^\ell m_i^2 - m_j \delta \right) $, and we are done, since the Singleton bound in this case is
$$ \log_q|\mathcal{C}| \leq \sum_{i=j}^\ell (t m_i)m_i - (t m_j) \delta = t \left( \sum_{i=j}^\ell m_i^2 - m_j \delta \right) . $$
\end{proof}

Consider now $ \ell \in [q-1] $ and let $ \mathcal{D} \subseteq \mathbb{F}_{q^m}^{\ell m} $ be an $ \mathbb{F}_{q^m} $-linear linearized Reed--Solomon code \cite{linearizedRS} (see also Subsection \ref{subsec preliminaries msrd}) of minimum sum-rank distance $ d \in [\ell m] $. Set $ \mathcal{C}_1 = \ldots = \mathcal{C}_t = M_{\boldsymbol\gamma}^{\mathbf{m}} (\mathcal{D}) \in (\mathbb{F}_q^{m \times m})^\ell $, in the cartesian-product construction from (\ref{eq def cartesian prod}), for $ \mathbf{m} = (m,\ldots, m) $ and for an ordered basis $ \boldsymbol\gamma = (\gamma_1, \ldots, \gamma_m) $ of $ \mathbb{F}_{q^m} $ over $ \mathbb{F}_q $. Then the code 
$$ \mathcal{C} \subseteq (\mathbb{F}_q^{(t m) \times m})^\ell $$
from (\ref{eq def cartesian prod}) is $ \mathbb{F}_q $-linear and MSRD of minimum sum-rank distance $ d $. 

The only MSRD codes with such parameters and with a known efficient decoder are linearized Reed--Solomon codes $ \mathcal{C}^\prime \subseteq \mathbb{F}_{q^{t m}}^{\ell m} \cong (\mathbb{F}_q^{(t m) \times m})^\ell $ of minimum sum-rank distance $ d $. However, decoding $ \mathcal{C} $ is always more efficient than decoding $ \mathcal{C}^\prime $, since $ \mathcal{C} $ requires decoding $ t $ linearized Reed--Solomon codes over $ \mathbb{F}_{q^m} $, $ \mathcal{C}^\prime $ requires decoding one linearized Reed--Solomon code over $ \mathbb{F}_{q^{t m}} $, in both cases of code length $ \ell m $, and there are no algorithms for multiplication in $ \mathbb{F}_{q^{t m}} $ of linear complexity (or lower) in $ t $ over $ \mathbb{F}_{q^m} $. 

For instance, if we use the Welch-Berlekamp decoder from \cite{secure-multishot}, then decoding $ \mathcal{C}^\prime $ requires $ \mathcal{O}((\ell m)^2) $ operations in $ \mathbb{F}_{q^{t m}} $, while decoding $ \mathcal{C} $ requires $ \mathcal{O}(t (\ell m)^2) $ operations in $ \mathbb{F}_{q^m} $. Assume that one multiplication in $ \mathbb{F}_{q^{t m}} $ costs about $ \mathcal{O}(t^2) $ operations in $ \mathbb{F}_{q^m} $. Then decoding $ \mathcal{C}^\prime $ requires $ \mathcal{O}((t \ell m)^2) $ operations in $ \mathbb{F}_{q^m} $, while decoding $ \mathcal{C} $ requires $ \mathcal{O}(t (\ell m)^2) $ operations in $ \mathbb{F}_{q^m} $.

\section{Construction 2: Combining bases} \label{sec combining}

Now we provide a construction that combines two linear codes by ``glueing'' their bases. 

\begin{construction} \label{cons bases}
Let
$$ \mathcal{C}_1 \subseteq \prod_{i=1}^\ell \mathbb{F}_q^{m_i \times n_i} \quad \textrm{and} \quad \mathcal{C}_2 \subseteq \prod_{i=1}^t \mathbb{F}_q^{m_{\ell+i} \times n_{\ell+i}} $$
be $ \mathbb{F}_q $-linear codes of dimensions $ k_1 $ and $ k_2 $, respectively. Set also $ d_1 = {\rm d}(\mathcal{C}_1) $ and $ d_2 = {\rm d}(\mathcal{C}_2) $.

Let $ \{ B_{j,1},\ldots, B_{j,k_j} \} $ form a basis of $ \mathcal{C}_j $, for $ j = 1,2 $. Consider the $ \mathbb{F}_q $-linear code $ \mathcal{C} \subseteq \prod_{i=1}^{\ell + t} \mathbb{F}_q^{m_i \times n_i} $ with basis
$$ \{ (B_{1,1},B_{2,1}), \ldots, (B_{1,k},B_{2,k}) \}, $$
where $ k = \min \{ k_1, k_2 \} $. 
\end{construction}

The code $ \mathcal{C} $ satisfies the following result, whose proof is straightforward.

\begin{lemma} \label{lemma combination}
It holds that 
$$ \dim (\mathcal{C}) = \min \{ k_1, k_2 \} \quad \textrm{and} \quad {\rm d}(\mathcal{C}) = d_1 + d_2. $$
\end{lemma}

Now assume that $ m_1 \geq \ldots \geq m_{\ell +t} $ and $ n_i \leq m_i $ for $ i \in [\ell + t] $. Assume also that $ \mathcal{C}_1 $ and $ \mathcal{C}_2 $ are MSRD with 
$$ d_1 = \sum_{i=1}^\ell n_i \quad \textrm{and} \quad d_2 = \sum_{i=\ell+1}^{j-1} n_i + \delta +1 , $$
for $ j \in [\ell+1, \ell +t] $ and $ 0 \leq \delta \leq m_j - 1 $. In particular, $ k_1 = m_\ell $ by the Singleton bound (\ref{eq singleton bound}). Finally, assume also that $ m_\ell \geq k_2 $. In this case, we have the following.

\begin{theorem} \label{th bases}
With assumptions as in the above paragraph, the code $ \mathcal{C} $ is MSRD with 
$$ {\rm d}(\mathcal{C}) = \sum_{i=1}^{j-1} n_i + \delta + 1 \quad \textrm{and} \quad \dim(\mathcal{C}) = \sum_{i=j}^{\ell +t} m_in_i - m_j \delta . $$
\end{theorem}
\begin{proof}
Trivial from Lemma \ref{lemma combination} and the parameters of $ \mathcal{C}_1 $ and $ \mathcal{C}_2 $.
\end{proof}

Observe that the main parameter restrictions are 
$$ {\rm d}(\mathcal{C}) > \sum_{i=1}^\ell n_i \quad \textrm{and} \quad m_\ell \geq \sum_{i=j}^{\ell +t} m_in_i - m_j \delta . $$
We also note that Construction \ref{cons bases} can be iterated any given number of times.

In Section \ref{sec comparisons}, we will show how Construction \ref{cons bases} generalizes constructions from the literature.

\section{Construction 3: Using lattices of MSRD codes} \label{sec lattices}

In this section, we provide a construction of $ \mathbb{F}_q $-linear MSRD codes based on lattices of (shorter) MSRD codes. We describe the general construction in Subsection \ref{subsec construction 3} and provide concrete examples in Subsection \ref{subsec examples 3}.

\subsection{The general construction} \label{subsec construction 3}

Consider the parameters $ m_1 \geq \ldots \geq m_\ell $ and $ n_i \leq m_i $ for $ i \in [\ell] $. We further assume that $ m = m_s = m_{s+1} = \ldots = m_\ell $, for some $ s \in [\ell] $. Set $ n = n_1 + \cdots + n_\ell $ and let $ d \in [n] $ be such that 
\begin{equation}
\label{eq condition d-t}
d - t \geq \sum_{i=1}^{s-1} n_i + 1, 
\end{equation}
for some positive integer $ t $. Consider an $ \mathbb{F}_q $-linear MSRD code $ \mathcal{C}_\varnothing \subseteq \prod_{i=1}^\ell \mathbb{F}_q^{m_i \times n_i} $ of distance $ {\rm d}(\mathcal{C}_\varnothing) = d $, let $ \{ B_{u,v} \}_{u=1,v=1}^{t,m} \subseteq \prod_{i=1}^\ell \mathbb{F}_q^{m_i \times n_i} $ be a set of $ \mathbb{F}_q $-linearly independent tuples such that $ \mathcal{C}_\varnothing \cap \langle B_{i,j} : i \in [t], j \in [m] \rangle_{\mathbb{F}_q} = 0 $, and define the $ \mathbb{F}_q $-linear code
\begin{equation}
\label{eq def code C_I}
\mathcal{C}_I = \mathcal{C}_\varnothing \oplus \langle B_{i,j} : i \in I, j \in [m] \rangle_{\mathbb{F}_q} ,
\end{equation}
for $ I \subseteq [t] $. Observe that this imposes the restriction $ tm + \dim_{\mathbb{F}_q}(\mathcal{C}_\varnothing) \leq \sum_{i=1}^\ell m_in_i $. Given $ I \subseteq [t] $, we have by definition that 
$$ \dim(\mathcal{C}_I) = \dim(\mathcal{C}_\varnothing) + m|I| = m (n - d + 1 + |I|). $$
We will further assume that $ {\rm d}(\mathcal{C}_I) = d - |I| $. This implies that $ \mathcal{C}_I $ is MSRD due to the Singleton bound (\ref{eq singleton bound}), since such a bound is $ m (n - d + 1 + |I|) $ in this case, since $ d-|I| \geq d-t \geq \sum_{i=1}^{s-1}n_i+1 $ by (\ref{eq condition d-t}), and $ m_s = \ldots = m_\ell = m $. Observe that the family $ \{ \mathcal{C}_I \}_{I \subseteq [t]} $ forms a lattice of MSRD codes isomorphic to the lattice of subsets of $ [t] $ by the map $ I \mapsto \mathcal{C}_I $. 

We now proceed to obtain a new $ \mathbb{F}_q $-linear MSRD code of distance $ d $ but longer than $ \mathcal{C}_\varnothing $. To that end, we consider additional lengths $ m_{\ell+1} , \ldots, m_{\ell + \ell_t}, n_{\ell+1} , \ldots, n_{\ell + \ell_t} $, for integers $ 0 = \ell_0 < \ell_1 < \ell_2 < \ldots < \ell_t $ such that
\begin{equation}
m_{\ell + \ell_{i-1} + 1} n_{\ell + \ell_{i-1} + 1} + \cdots + m_{\ell + \ell_i} n_{\ell + \ell_i} \leq m, 
\label{eq condition m bigger than squares}
\end{equation}
for $ i \in [t] $. Consider now $ \mathbb{F}_q $-linear subspaces $ \mathcal{V}_j \subseteq \mathbb{F}_q^m $ such that $ \dim (\mathcal{V}_j) = m_{\ell + j} n_{\ell + j} $, for $ j \in [\ell_t] $, and such that
$$ \mathcal{V}_{\ell_{i-1} + 1} , \mathcal{V}_{\ell_{i-1} + 2} , \ldots , \mathcal{V}_{\ell_i} $$
form a direct sum inside $ \mathbb{F}_q^m $, for $ i \in [t] $. This is possible thanks to condition (\ref{eq condition m bigger than squares}). Finally, consider $ \mathbb{F}_q $-linear vector space isomorphisms 
$$ \varphi_j : \mathcal{V}_j \longrightarrow \mathbb{F}_q^{m_{\ell + j} \times n_{\ell + j}}, $$
for $ j \in [\ell_t] $.

The main construction of this section is as follows.

\begin{construction} \label{cons msrd}
We construct the $ \mathbb{F}_q $-linear code $ \mathcal{C} \subseteq \prod_{i=1}^{\ell + \ell_t} \mathbb{F}_q^{m_i \times n_i} $ as a direct sum of two subcodes $ \mathcal{C}_1 $ and $ \mathcal{C}_2 $. First, let $ \mathcal{C}_1 \subseteq \prod_{i=1}^{\ell + \ell_t} \mathbb{F}_q^{m_i \times n_i} $ be equal to $ \mathcal{C}_\varnothing $ but adding zeros to each codeword in the $ i $th block for every $ i \in [\ell + 1, \ell + \ell_t] $. Second, let 
$$ \mathcal{C}_2 = \bigoplus_{i=1}^t \bigoplus_{j=\ell_{i-1}+1}^{\ell_i} \left\lbrace \left( \sum_{k=1}^m \alpha_k B_{i,k} , 0, \ldots, \underbrace{\varphi_j(\boldsymbol\alpha)}_{(\ell + j)\textrm{th block}} , \ldots , 0 \right) : \boldsymbol\alpha \in \mathcal{V}_j \right\rbrace \subseteq \prod_{i=1}^{\ell + \ell_t} \mathbb{F}_q^{m_i \times n_i} , $$
where we use the notation $ \boldsymbol\alpha = (\alpha_1, \ldots, \alpha_m) \in \mathbb{F}_q^m $. Finally, define $ \mathcal{C} = \mathcal{C}_1 \oplus \mathcal{C}_2 $.
\end{construction}

We next show that the code $ \mathcal{C} $ is an $ \mathbb{F}_q $-linear MSRD code of minimum distance $ d $.

\begin{theorem} \label{th extension is msrd}
The code $ \mathcal{C} $ from Construction \ref{cons msrd} is an $ \mathbb{F}_q $-linear MSRD code of minimum sum-rank distance $ {\rm d}(\mathcal{C}) = d $ and dimension $ \dim_{\mathbb{F}_q}(\mathcal{C}) = m (n-d+1) + \sum_{i=\ell+1}^{\ell + \ell_t} m_in_i $.
\end{theorem}
\begin{proof}
First, let
$$ \mathcal{D}_j = \left\lbrace \left( \sum_{k=1}^m \alpha_k B_{i,k} , 0, \ldots, \underbrace{\varphi_j(\boldsymbol\alpha)}_{(\ell + j)\textrm{th block}} , \ldots , 0 \right) : \boldsymbol\alpha \in \mathcal{V}_j \right\rbrace \subseteq \prod_{i=1}^{\ell + \ell_t} \mathbb{F}_q^{m_i \times n_i}, $$
for $ j \in [\ell_{i-1} + 1 , \ell_i] $ and $ i \in [t] $. Clearly $ \mathcal{D}_j $ is an $ \mathbb{F}_q $-linear subspace isomorphic to $ \mathcal{V}_j $ and thus of dimension $ m_{\ell + j} n_{\ell + j} $. Observe now that all the subspaces
$$ \mathcal{D}_1 , \mathcal{D}_2 , \ldots , \mathcal{D}_{\ell_t} $$
form a direct sum inside $ \prod_{i=1}^{\ell + \ell_t} \mathbb{F}_q^{m_i \times n_i} $, since a nonzero codeword in $ \mathcal{D}_j $ has a nonzero component in the $ (\ell + j) $th block for some $ j \in [\ell_{i-1}+1, \ell_i] $, for some $ i \in [t] $, and is identically zero in all the other rank blocks with indices in $ [\ell + 1, \ell + \ell_t] $. Therefore we indeed have that
$$ \mathcal{C}_2 = \bigoplus_{i=1}^t \bigoplus_{j=\ell_{i-1}+1}^{\ell_i} \mathcal{D}_j. $$
In particular, we have that
$$ \dim(\mathcal{C}_2) = \sum_{i=1}^t \sum_{j=\ell_{i-1}+1}^{\ell_i} \dim(\mathcal{D}_j) = \sum_{i=\ell+1}^{\ell + \ell_t} m_in_i. $$

Similarly, since every nonzero codeword in $ \mathcal{C}_2 $ contains a nonzero element in at least one of the blocks in the positions $ j \in [\ell+1,\ell+\ell_t] $ and $ \mathcal{C}_1 $ is identically zero in those positions, we also deduce that $ \mathcal{C}_1 \cap \mathcal{C}_2 = 0 $. In particular, it holds indeed that $ \mathcal{C} = \mathcal{C}_1 \oplus \mathcal{C}_2 $, and
$$ \dim(\mathcal{C}) = \dim(\mathcal{C}_1) + \dim(\mathcal{C}_2) = m (n-d+1) + \sum_{i=\ell+1}^{\ell + \ell_t} m_in_i. $$

Now we show that the minimum distance of $ \mathcal{C} $ is $ d $. A codeword in $ \mathcal{C} $ is of the form
$$ C = \left( D + \sum_{i=1}^t \sum_{j = \ell_{i-1}+1}^{\ell_i} \sum_{k=1}^m \alpha_{j,k} B_{i,k}, \varphi_1(\boldsymbol\alpha_1), \ldots, \varphi_{\ell_t}(\boldsymbol\alpha_{\ell_t}) \right), $$
where $ D \in \mathcal{C}_\varnothing $ and $ \boldsymbol\alpha_j = (\alpha_{j,1}, \ldots, \alpha_{j,m}) \in \mathcal{V}_j $, for $ j \in [\ell_t] $. Set
$$ I = \{ i \in [t] \mid \exists j \in [\ell_{i-1}+1,\ell_i] \textrm{ such that } \boldsymbol\alpha_j \neq \mathbf{0} \}. $$
Then we have
$$ C = \left( D + \sum_{i \in I} \sum_{j = \ell_{i-1}+1}^{\ell_i} \sum_{k=1}^m \alpha_{j,k} B_{i,k} , \varphi_1(\boldsymbol\alpha_1), \ldots, \varphi_{\ell_t}(\boldsymbol\alpha_{\ell_t}) \right). $$
On the first $ \ell $ blocks, we have the codeword
\begin{equation}
D + \sum_{i \in I} \sum_{j = \ell_{i-1}+1}^{\ell_i} \sum_{k=1}^m \alpha_{j,k} B_{i,k} \in \mathcal{C}_I.
\label{eq codeword}
\end{equation}
Given $ i \in I $, observe that $ \sum_{k=1}^m \left( \sum_{j = \ell_{i-1}+1}^{\ell_i} \alpha_{j,k} \right) B_{i,k} \neq 0 $, since $ B_{i,1},\ldots, B_{i,m} $ are $ \mathbb{F}_q $-linearly independent, $ \mathcal{V}_{\ell_{i-1}+1} , \ldots, \mathcal{V}_{\ell_i} $ form a direct sum inside $ \mathbb{F}_q^m $ and there is at least one $ j \in [\ell_{i-1}+1, \ell_i] $ such that $ \boldsymbol\alpha_j \neq \mathbf{0} $. In particular, $ \sum_{i \in I} \left( \sum_{j = \ell_{i-1}+1}^{\ell_i} \sum_{k=1}^m \alpha_{j,k} B_{i,k} \right) \neq 0 $ since $ \{ B_{i,j} \}_{i=1,j=1}^{t,m} $ are $ \mathbb{F}_q $-linearly independent. Combining this fact with $ \mathcal{C}_\varnothing \cap \langle B_{i,j} : i \in [t], j \in [m] \rangle_{\mathbb{F}_q} = 0 $, we conclude that the codeword in (\ref{eq codeword}) is zero if, and only if, $ D = 0 $ and $ I = \varnothing $, which is equivalent to $ C $ being zero. Hence if $ C $ is nonzero, then
$$ {\rm wt} \left( D + \sum_{i \in I} \sum_{j = \ell_{i-1}+1}^{\ell_i} \sum_{k=1}^m \alpha_{j,k} B_{i,k} \right) \geq {\rm d}(\mathcal{C}_I) = d - |I| . $$
Finally, since there is at least one $ j \in [\ell_{i-1}+1,\ell_i] $ such that $ \boldsymbol\alpha_j \neq \mathbf{0} $, for every $ i \in I $, then 
$$ {\rm wt}(\varphi_1(\boldsymbol\alpha_1), \ldots, \varphi_{\ell_t}(\boldsymbol\alpha_{\ell_t})) \geq |I|, $$
and we conclude that $ {\rm wt}(C) \geq d $ if $ C $ is nonzero. In other words, $ {\rm d}(\mathcal{C}) \geq d $, but equality must hold by the Singleton bound (\ref{eq singleton bound}), thus $ {\rm d}(\mathcal{C}) = d $ and we are done.
\end{proof}

\subsection{Concrete examples} \label{subsec examples 3}

Lattices of MSRD codes were studied in \cite{doubly} in order to extend the MSRD codes from \cite{generalMSRD}, i.e., those from Subsection \ref{subsec preliminaries msrd}. However, the extensions from \cite{doubly} only added blocks of matrices of size $ 1 \times m $. Using the technique from Subsection \ref{subsec construction 3}, we now give extensions of the MSRD codes from Subsection \ref{subsec preliminaries msrd} for new ranges of parameters, providing new constructions of MSRD codes. 

Consider $ m = m_1 = \ldots = m_\ell $ and $ r = n_1 = \ldots = n_\ell \leq m $ for $ i \in [\ell] $, and set $ n = \ell r $. Let $ k $ and $ t $ be positive integers such that $ t+k \leq n $ and let $ \mathbf{g}_1, \mathbf{g}_2, \ldots , \mathbf{g}_{t+k} \in \mathbb{F}_{q^m}^n $ be $ \mathbb{F}_{q^m} $-linearly independent. For $ I \subseteq [t] $, define the $ \mathbb{F}_{q^m} $-linear code $ \mathcal{D}_I = \langle \mathbf{g}_i : i \in I \rangle_{\mathbb{F}_{q^m}} \oplus \langle \mathbf{g}_{t+1}, \ldots, \mathbf{g}_{t+k} \rangle_{\mathbb{F}_{q^m}} \subseteq \mathbb{F}_{q^m}^n $, and assume that it is MSRD, that is, 
$$ \dim_{\mathbb{F}_{q^m}}(\mathcal{D}_I) = k + |I| \quad \textrm{and} \quad {\rm d}(\mathcal{D}_I) = n - k - |I| + 1. $$
If $ \boldsymbol\gamma = (\gamma_1, \ldots, \gamma_m) $ forms an ordered basis of $ \mathbb{F}_{q^m} $ over $ \mathbb{F}_q $ and we define $ \mathcal{C}_I = M_{\boldsymbol\gamma}^{\mathbf{n}}(\mathcal{D}_I) \subseteq \prod_{i=1}^\ell \mathbb{F}_q^{m \times n_i} $, then $ \{ \mathcal{C}_I \}_{I \subseteq [t]} $ forms a lattice of $ \mathbb{F}_q $-linear MSRD codes as in Subsection \ref{subsec construction 3}, where $ {\rm d}(\mathcal{C}_\varnothing) = d = n-k+1 $ and $ {\rm d}(\mathcal{C}_I) = d - |I| $, for $ I \subseteq [t] $. In Construction \ref{cons msrd}, we set $ B_{i,j} = M_{\boldsymbol\gamma}^{\mathbf{n}}(\gamma_j \mathbf{g}_i) $, for $ i \in [t] $ and $ j \in [m] $, and the condition $ tm + \dim_{\mathbb{F}_q}(\mathcal{C}_\varnothing) \leq mn $ is satisfied. Note also that we may take $ s = 1 $ since $ m_1 = \ldots = m_\ell = m $ and $ d-t \geq 1 $.

When $ t = 2 $, one way of constructing the vectors $ \mathbf{g}_1, \mathbf{g}_2, \ldots , \mathbf{g}_{t+k} \in \mathbb{F}_{q^m}^n $ is as follows. Consider
$$ \left( \begin{array}{c}
\mathbf{g}_1 \\
\hline
\mathbf{g}_3 \\
\mathbf{g}_4 \\
\vdots \\
\mathbf{g}_{k+2} \\
\hline
\mathbf{g}_{2}
\end{array} \right) = \left( \begin{array}{lll|c|lll}
\beta_1 & \ldots & \beta_{\mu r} & \ldots & \beta_1 & \ldots & \beta_{\mu r} \\
 \beta_1^q a_1 & \ldots & \beta_{\mu r}^q a_1 & \ldots & \beta_1^q a_{q-1} & \ldots & \beta_{\mu r}^q a_{q-1} \\
\beta_1^{q^2} a_1^{\frac{q^2-1}{q-1}} & \ldots & \beta_{\mu r}^{q^2} a_1^{\frac{q^2-1}{q-1}} & \ldots &  \beta_1^{q^2} a_{q-1}^{\frac{q^2-1}{q-1}} & \ldots &  \beta_{\mu r}^{q^2} a_{q-1}^{\frac{q^2-1}{q-1}} \\
\vdots & \ddots & \vdots & \ddots & \vdots & \ddots & \vdots \\
 \beta_1^{q^{k}} a_1^{\frac{q^{k}-1}{q-1}} & \ldots &  \beta_{\mu r}^{q^{k}} a_1^{\frac{q^{k}-1}{q-1}} & \ldots &  \beta_1^{q^{k}} a_{q-1}^{\frac{q^{k}-1}{q-1}} & \ldots &  \beta_{\mu r}^{q^{k}} a_{q-1}^{\frac{q^{k}-1}{q-1}} \\
 \beta_1^{q^{k+1}} a_1^{\frac{q^{k+1}-1}{q-1}} & \ldots &  \beta_{\mu r}^{q^{k+1}} a_1^{\frac{q^{k+1}-1}{q-1}} & \ldots &  \beta_1^{q^{k+1}} a_{q-1}^{\frac{q^{k+1}-1}{q-1}} & \ldots &  \beta_{\mu r}^{q^{k+1}} a_{q-1}^{\frac{q^{k+1}-1}{q-1}} \\
\end{array} \right) , $$
where $ \ell = \mu (q-1) $, $ n = \ell r $, and $ a_1, \ldots, a_{q-1}, \beta_1, \ldots, \beta_{\mu r} \in \mathbb{F}_{q^m}^* $ satisfy the properties stated after equation (\ref{eq general msrd generator}).
%
With these assumptions, $ \mathbf{g}_1, \mathbf{g}_2, \mathbf{g}_3 \ldots, \mathbf{g}_{k+2} \in \mathbb{F}_{q^m}^n $ are $ \mathbb{F}_{q^m} $-linearly independent and the $ \mathbb{F}_{q^m} $-linear codes $ \mathcal{D}_I = \langle \mathbf{g}_i : i \in I \rangle_{\mathbb{F}_{q^m}} \oplus \langle \mathbf{g}_{3}, \ldots, \mathbf{g}_{k+2} \rangle_{\mathbb{F}_{q^m}} \subseteq \mathbb{F}_{q^m}^n $, for $ I \subseteq \{1,2 \} $, are MSRD by \cite[Th. 3.12]{generalMSRD} and \cite[Lemma 5]{doubly}.
%

In \cite[Cor. 8]{doubly}, it was shown how to extend these MSRD codes by adding $ t=2 $ rank blocks each formed by matrices of sizes $ 1 \times m $ (i.e., adding a Hamming-metric block $ \mathbb{F}_{q^m}^2 $). With Construction \ref{cons msrd}, we may extend them to obtain an $ \mathbb{F}_q $-linear MSRD code $ \mathcal{C} \subseteq \prod_{i=1}^{\ell + \ell_2} \mathbb{F}_q^{m_i \times n_i} $ with $ {\rm d}(\mathcal{C}) = d $ by adding $ t=2 $ sets of blocks of any sizes $ m_{\ell + 1} \times n_{\ell + 1}, \ldots, m_{\ell + \ell_2} \times n_{\ell + \ell_2} $, with the only restrictions
\begin{equation*}
\begin{split}
m_{\ell + 1} \times n_{\ell + 1} + \cdots + m_{\ell + \ell_1} \times n_{\ell + \ell_1} \leq 1 \times m, & \\
m_{\ell + \ell_1 + 1} \times n_{\ell + \ell_1 + 1} + \cdots + m_{\ell + \ell_2} \times n_{\ell + \ell_2} \leq 1 \times m, &
\end{split}
\end{equation*}
where $ 0 < \ell_1 < \ell_2 $, hence achieving more flexibility in how we may extend such MSRD codes. In particular, the extension may be obtained by adding a block with a sum-rank metric that is not the Hamming metric, in contrast with \cite{doubly}. This is the first known extension of the MSRD codes from \cite{generalMSRD} by adding rank blocks of matrices of sizes different than $ 1 \times m $.

In \cite[Sec. 7]{doubly}, the MSRD extension as above adding a Hamming-metric block $ \mathbb{F}_{q^m}^2 $ was shown to be a one-weight code in some cases (that is, a code whose nonzero codewords all have the same sum-rank weight). The same result holds for the general code $ \mathcal{C} $ as above. The following proposition is straightforward by \cite[Prop. 13]{doubly}. 

\begin{proposition}
Let $ \mathcal{C} \subseteq \prod_{i=1}^{\ell + \ell_2} \mathbb{F}_q^{m_i \times n_i} $ be as above and assume that $ \dim_{\mathbb{F}_q}(\mathcal{C}) = 2m $. Then $ \mathcal{C} $ is a one-weight code if, and only if, $ \ell_1 = 1 $, $ \ell_2 = 2 $ and $ \bigcup_{i=1}^\mu \mathcal{H}_i = \mathbb{F}_{q^m} $, where $ \mathcal{H}_1, \ldots, \mathcal{H}_\mu $ are as in Subsection \ref{subsec preliminaries msrd}.
\end{proposition}

A family of lattices of MSRD codes for $ t = 3 $ can be obtained as follows, although only for $ k = 0 $ (i.e., $ \mathcal{D}_\varnothing = 0 $), $ m $ odd and $ q $ even. Consider
\begin{equation}
\left( \begin{array}{c}
\mathbf{g}_1 \\
\mathbf{g}_2 \\
\mathbf{g}_3 \\
\end{array} \right) = \left( \begin{array}{lll|c|lll}
\beta_1 & \ldots & \beta_{\mu r} & \ldots & \beta_1 & \ldots & \beta_{\mu r} \\
 \beta_1^q a_1 & \ldots & \beta_{\mu r}^q a_1 & \ldots & \beta_1^q a_{q-1} & \ldots & \beta_{\mu r}^q a_{q-1} \\
\beta_1^{q^2} a_1^{q+1} & \ldots & \beta_{\mu r}^{q^2} a_1^{q+1} & \ldots &  \beta_1^{q^2} a_{q-1}^{q+1} & \ldots &  \beta_{\mu r}^{q^2} a_{q-1}^{q+1} \\
\end{array} \right) ,
\label{eq lattice from LRS 2}
\end{equation}
where $ \ell = \mu (q-1) $, $ n = \ell r $, and $ a_1, \ldots, a_{q-1}, \beta_1, \ldots, \beta_{\mu r} \in \mathbb{F}_{q^m}^* $ satisfy the properties stated after equation (\ref{eq general msrd generator}). If we further assume that $ m $ is odd and $ q $ is even, then it was shown in the proof of \cite[Th. 5]{doubly} that $ \mathbf{g}_1, \mathbf{g}_2, \mathbf{g}_3 \in \mathbb{F}_{q^m}^n $ are $ \mathbb{F}_{q^m} $-linearly independent and $ \mathcal{D}_I = \langle \mathbf{g}_i : i \in I \rangle_{\mathbb{F}_{q^m}} \subseteq \mathbb{F}_{q^m}^n $, for $ I \subseteq \{ 1,2,3 \} $, are MSRD. Notice that in this case $ \mathcal{D}_\varnothing = 0 $, $ d = n+1 $ and $ {\rm d}(\mathcal{D}_I) = d - |I| = n+1-|I| $, for $ I \subseteq \{ 1,2,3 \} $.

In \cite[Th. 3]{doubly}, it was shown how to extend these MSRD codes by adding $ t=3 $ rank blocks each formed by matrices of sizes $ 1 \times m $ (i.e., adding a Hamming-metric block $ \mathbb{F}_{q^m}^3 $). With Construction \ref{cons msrd}, we may extend them to obtain an $ \mathbb{F}_q $-linear MSRD code $ \mathcal{C} \subseteq \prod_{i=1}^{\ell + \ell_3} \mathbb{F}_q^{m_i \times n_i} $ with $ {\rm d}(\mathcal{C}) = d $ by adding $ t=3 $ sets of blocks of any sizes $ m_{\ell + 1} \times n_{\ell + 1}, \ldots, m_{\ell + \ell_3} \times n_{\ell + \ell_3} $, with the only restrictions
\begin{equation*}
\begin{split}
m_{\ell + 1} \times n_{\ell + 1} + \cdots + m_{\ell + \ell_1} \times n_{\ell + \ell_1} \leq 1 \times m, & \\
m_{\ell + \ell_1 + 1} \times n_{\ell + \ell_1 + 1} + \cdots + m_{\ell + \ell_2} \times n_{\ell + \ell_2} \leq 1 \times m, & \\
m_{\ell + \ell_2 + 1} \times n_{\ell + \ell_2 + 1} + \cdots + m_{\ell + \ell_3} \times n_{\ell + \ell_3} \leq 1 \times m, &
\end{split}
\end{equation*}
where $ 0 < \ell_1 < \ell_2 < \ell_3 $, hence achieving more flexibility in how we may extend such MSRD codes, as in the case $ t = 2 $ shown earlier.

%
%
%
%

\section{Construction 4: Using systematic MSRD codes} \label{sec systematic}

In this section, we provide a construction of $ \mathbb{F}_q $-linear MSRD codes based on systematic generator matrices of $ \mathbb{F}_{q^m} $-linear MSRD codes in $ \mathbb{F}_{q^m}^n $. We describe the general construction in Subsection \ref{subsec construction 4} and provide concrete examples in Subsections \ref{subsec examples 4 phi} and \ref{subsec examples 4 systematic}.

\subsection{The general construction} \label{subsec construction 4}

Consider the parameters $ m = m_1 = \ldots = m_\ell $ and $ n_i \leq m $, for $ i \in [\ell] $. Let also $ t \in [m] $, define $ n = n_1 + \cdots + n_\ell $ and let $ \mathcal{D}_0 \subseteq \mathbb{F}_{q^m}^{n+t} $ be an $ \mathbb{F}_{q^m} $-linear MSRD code of distance $ {\rm d}(\mathcal{D}_0) = d-t \geq 1 $, for some $ d \in [t+1,t+n] $, for the sum-rank length partition $ (n_1, \ldots, n_\ell, t) $. Hence $ \dim_{\mathbb{F}_{q^m}}(\mathcal{D}_0) = n-d+1+2t $. We will set $ k = n+t-d+1 $. Consider a generator matrix of $ \mathcal{D}_0 $ of the form
\begin{equation}
G_0 = \left( \begin{array}{c|cccc}
\mathbf{g}_1 & 1 & 0 & \ldots & 0 \\
\mathbf{g}_2 & 0 & 1 & \ldots & 0 \\
\vdots & \vdots & \vdots & \ddots & \vdots \\
\mathbf{g}_t & 0 & 0 & \ldots & 1 \\
\hline
\mathbf{g}_{t+1} & 0 & 0 & \ldots & 0 \\
\vdots & \vdots & \vdots & \ddots & \vdots \\
\mathbf{g}_{t+k} & 0 & 0 & \ldots & 0
\end{array} \right) \in \mathbb{F}_{q^m}^{(t+k) \times (n+t)},
\label{eq systematic gen matrix}
\end{equation}
%
%
where $ \mathbf{g}_1, \ldots, \mathbf{g}_{t+k} \in \mathbb{F}_{q^m}^n $. Such a generator matrix exists by Gaussian elimination and the fact that the last $ \dim_{\mathbb{F}_{q^m}}(\mathcal{D}_0) \geq t $ positions form an information set of $ \mathcal{D}_0 $ since it is MSRD, thus MDS (see \cite[Ch. 1]{fnt}). Notice that $ G_0 $ is only a systematic generator matrix if $ k = 0 $. However, we will still call it systematic for simplicity.

Assume that there is an $ \mathbb{F}_q $-linear subspace $ \mathcal{V} \subseteq \mathbb{F}_{q^m}^t $ and a vector space isomorphism 
\begin{equation}
\phi : \mathcal{V} \longrightarrow \prod_{i=\ell+1}^{\ell+u} \mathbb{F}_q^{m_i \times n_i},
\label{eq map phi}
\end{equation}
for positive integers $ u $, $ m \geq m_{\ell + 1} \geq \ldots \geq m_{\ell + u} $ and $ n_i \leq m_i $, for $ i \in [\ell +1 , \ell + u] $, such that 
\begin{equation}
{\rm wt}(\phi(\boldsymbol\lambda)) \geq {\rm wt}(\boldsymbol\lambda),
\label{eq property of isomorphism}
\end{equation}
for all $ \boldsymbol\lambda \in \mathcal{V} $. We will provide examples of such an isomorphism in Subsection \ref{subsec examples 4 phi}. Notice that a necessary condition for its existence is
$$ tm \geq m_{\ell+1} n_{\ell+1} + \cdots + m_{\ell +u}n_{\ell + u}. $$
The main construction of this section is as follows.

\begin{construction} \label{cons msrd systematic}
Fix an ordered basis $ \boldsymbol\gamma \in \mathbb{F}_{q^m}^m $ of $ \mathbb{F}_{q^m} $ over $ \mathbb{F}_q $, set $ \mathbf{n} = (n_1, \ldots, n_\ell) $ and define
$$ \mathcal{C} = \left\lbrace \left( M_{\boldsymbol\gamma}^\mathbf{n} \left( \sum_{i=1}^{t+k} \lambda_i \mathbf{g}_i \right) , \phi(\lambda_1, \ldots, \lambda_t) \right) : (\lambda_1, \ldots, \lambda_t) \in \mathcal{V}, \lambda_{t+1} , \ldots, \lambda_{t+k} \in \mathbb{F}_{q^m} \right\rbrace \subseteq \prod_{i=1}^{\ell + u} \mathbb{F}_q^{m_i \times n_i}. $$
\end{construction}

We next show that the code $ \mathcal{C} $ is an $ \mathbb{F}_q $-linear MSRD code of minimum distance $ d $.

\begin{theorem} \label{th extension is msrd systematic}
The code $ \mathcal{C} $ from Construction \ref{cons msrd systematic} is an $ \mathbb{F}_q $-linear MSRD code of minimum sum-rank distance $ {\rm d}(\mathcal{C}) = d $ and dimension $ \dim_{\mathbb{F}_q}(\mathcal{C}) = m (n-d+1) + \sum_{i=\ell+1}^{\ell + u} m_in_i $.
\end{theorem}
\begin{proof}
Similarly to Construction \ref{cons msrd} and Theorem \ref{th extension is msrd}, we may write the code as the direct sum $ \mathcal{C} = \mathcal{C}_1 \oplus \mathcal{C}_2 $, where 
$$ \mathcal{C}_1 = M_{\boldsymbol\gamma}^\mathbf{n} \left( \langle \mathbf{g}_{t+1} ,\ldots, \mathbf{g}_{t+k} \rangle_{\mathbb{F}_{q^m}} \right) \times 0 , $$
where $ 0 $ is the zero subspace in $ \prod_{i=\ell+1}^{\ell + u} \mathbb{F}_q^{m_i \times n_i} $, and
$$ \mathcal{C}_2 = \left\lbrace \left( M_{\boldsymbol\gamma}^\mathbf{n} \left( \sum_{i=1}^t \lambda_i \mathbf{g}_i \right) , \phi(\lambda_1, \ldots, \lambda_t) \right) : (\lambda_1, \ldots, \lambda_t) \in \mathcal{V} \right\rbrace . $$
It holds that $ \mathcal{C}_1 \cap \mathcal{C}_2 = 0 $, since any nonzero codeword in $ \mathcal{C}_2 $ has a nonzero component in at least one of the last $ u $ rank blocks, whereas $ \mathcal{C}_1 $ is identically zero in such positions. Thus $ \mathcal{C} = \mathcal{C}_1 \oplus \mathcal{C}_2 $. Next, the claim on the dimension of $ \mathcal{C} $ follows from the fact that $ \dim_{\mathbb{F}_q}(\mathcal{C}_1) = m(n-d+1) $ and
$$ \dim_{\mathbb{F}_q}(\mathcal{C}_2) = \dim_{\mathbb{F}_q}(\mathcal{V}) = \sum_{i=\ell+1}^{\ell + u} m_in_i, $$
since $ \phi $ is a vector space isomorphism.

Now let 
$$ C = \left( M_{\boldsymbol\gamma}^\mathbf{n} \left( \sum_{i=1}^{t+k} \lambda_i \mathbf{g}_i \right) , \phi(\boldsymbol\lambda) \right) \in \mathcal{C} \setminus 0, $$
for $ \lambda_1, \ldots, \lambda_{t+k} \in \mathbb{F}_{q^m} $, where $ \boldsymbol\lambda = (\lambda_1, \ldots, \lambda_t) \in \mathcal{V} $. We have that
$$ \mathbf{c} = \left( \sum_{i=1}^{t+k} \lambda_i \mathbf{g}_i, \boldsymbol\lambda \right) \in \mathcal{D}_0, $$
which is nonzero since $ C $ is nonzero. Finally, we have that
$$ {\rm wt}(C) = {\rm wt}\left( M_{\boldsymbol\gamma}^\mathbf{n} \left( \sum_{i=1}^{t+k} \lambda_i \mathbf{g}_i \right) \right) + {\rm wt} \left( \phi(\boldsymbol\lambda) \right) $$
$$ \geq {\rm wt}\left( \sum_{i=1}^{t+k} \lambda_i \mathbf{g}_i \right) + {\rm wt}(\boldsymbol\lambda) = {\rm wt}(\mathbf{c}) \geq {\rm d}(\mathcal{D}_0) = d, $$
where the first inequality holds by (\ref{eq property of isomorphism}). Therefore, $ {\rm d}(\mathcal{C}) \geq d $, and by the Singleton bound (\ref{eq singleton bound}), equality must hold.
\end{proof}

\subsection{Concrete examples for the isomorphism $ \phi $} \label{subsec examples 4 phi}

We start with a construction of the map $ \phi $ from (\ref{eq map phi}), i.e., a construction of an $ \mathbb{F}_q $-linear subspace $ \mathcal{V} \subseteq \mathbb{F}_{q^m}^t $ and a vector space isomorphism $ \phi : \mathcal{V} \longrightarrow \prod_{i= \ell+1}^{\ell + u} \mathbb{F}_q^{m_i \times n_i} $ such that $ {\rm wt}(\phi(\boldsymbol\lambda)) \geq {\rm wt}(\boldsymbol\lambda) $, for all $ \boldsymbol\lambda \in \mathbb{F}_{q^m}^t $. The idea will be to partition matrices into disjoint submatrices.

\begin{definition}
Given $ X \subseteq [m] $ and $ Y \subseteq [t] $, define $ \pi_{X,Y} : \mathbb{F}_q^{m \times t} \longrightarrow \mathbb{F}_q^{|X| \times |Y|} $ as the map such that $ \pi_{X,Y}(C) $ is the submatrix of $ C \in \mathbb{F}_q^{m \times t} $ formed by its entries in the positions $ (i,j) \in X \times Y $.
\end{definition}

\begin{definition} \label{def pi}
Consider $ X_1, \ldots , X_u \subseteq [m] $ and $ Y_1, \ldots, Y_u \subseteq [t] $ such that $ (X_i \times Y_i) \cap (X_j \times Y_j) = \varnothing $ if $ i \neq j $. Next, define the surjective $ \mathbb{F}_q $-linear map $ \pi : \mathbb{F}_q^{m \times t} \longrightarrow \prod_{i= \ell+1}^{\ell + u} \mathbb{F}_q^{m_i \times n_i} $ by
$$ \pi (C) = \left( \pi_{X_1,Y_1}(C), \ldots, \pi_{X_u,Y_u}(C) \right), $$
for $ C \in \mathbb{F}_q^{m \times t} $.
\end{definition}

We illustrate this definition with the following example.

\begin{example}
Consider the case $ m = 4 $, $ t = 5 $ and $ u = 5 $, and choose the following partition
$$ \begin{array}{rclcrcl}
X_1 & = & \{ 1,2,3 \}, & \quad & Y_1 & = & \{ 1,2,3 \}, \\
X_2 & = & \{ 4 \}, & \quad & Y_2 & = & \{ 1,2 \}, \\
X_3 & = & \{ 1 \}, & \quad & Y_3 & = & \{ 4,5 \}, \\
X_4 & = & \{ 2,3 \}, & \quad & Y_4 & = & \{ 4,5 \}, \\
X_5 & = & \{ 4 \}, & \quad & Y_5 & = & \{ 3,4,5 \}. 
\end{array} $$
Observe that $ (X_i \times Y_i) \cap (X_j \times Y_j) = \varnothing $ if $ i \neq j $. Now, the map 
$$ \pi : \mathbb{F}_q^{4 \times 5} \longrightarrow \prod_{i= 1}^5 \mathbb{F}_q^{|X_i| \times |Y_i|} $$
from Definition \ref{def pi} essentially consists in partitioning a matrix from $ \mathbb{F}_q^{4 \times 5} $ as follows:
$$ 
\begin{tikzpicture}[decoration=brace]
  \matrix (m)[
    matrix of math nodes,
    left delimiter=(,right delimiter={)},
    nodes in empty cells,
  ] {
    c_{1,1} & c_{1,2}  & c_{1,3} & c_{1,4}  & c_{1,5} \\
    c_{2,1} & c_{2,2}  & c_{2,3} & c_{2,4}  & c_{2,5} \\
    c_{3,1} & c_{3,2}  & c_{3,3} & c_{3,4}  & c_{3,5} \\
    c_{4,1} & c_{4,2}  & c_{4,3} & c_{4,4}  & c_{4,5} \\
  } ;

  \draw (m-3-1.south west) rectangle (m-1-3.north east);
  \draw (m-4-1.south west) rectangle (m-4-2.north east);
  \draw (m-1-4.south west) rectangle (m-1-5.north east);
  \draw (m-3-4.south west) rectangle (m-2-5.north east);
  \draw (m-4-3.south west) rectangle (m-4-5.north east);
\end{tikzpicture}
. $$
In this example, each set $ X_i $ consists of consecutive numbers in $ [m] $, and similarly for the sets $ Y_i $. Furthermore, in this example $ [m] \times [t] = \bigcup_{i=1}^5 X_i \times Y_i $. However, these two properties do not need to hold according to Definition \ref{def pi}.
\end{example}

Let the notation and assumptions be as in Definition \ref{def pi}. By the well-known properties of ranks of matrices and their submatrices, it holds that
\begin{equation}
{\rm Rk}(C) \leq \sum_{i=1}^u {\rm Rk}(\pi_{X_i,Y_i}(C)),
\label{eq inequality ranks for phi}
\end{equation}
for all $ C \in \mathbb{F}_q^{m \times t} $. Therefore, we may define the map $ \phi $ and the subspace $ \mathcal{V} $ as follows.

\begin{definition} \label{def phi}
Consider $ X_1, \ldots , X_u \subseteq [m] $ and $ Y_1, \ldots, Y_u \subseteq [t] $ such that $ (X_i \times Y_i) \cap (X_j \times Y_j) = \varnothing $ if $ i \neq j $. Let $ \boldsymbol\gamma = (\gamma_1, \ldots, \gamma_m) $ be an ordered basis of $ \mathbb{F}_{q^m} $ over $ \mathbb{F}_q $, and set
$$ \mathcal{U} = \left\lbrace (c_{i,j})_{i=1,j=1}^{m,t} \in \mathbb{F}_q^{m \times t} : c_{i,j} = 0, \textrm{ for } (i,j) \in ([m] \times [t]) \setminus \bigcup_{s=1}^u (X_s \times Y_s) \right\rbrace . $$
Finally, define $ \mathcal{V} = (M_{\boldsymbol\gamma}^t)^{-1}( \mathcal{U} ) \subseteq \mathbb{F}_{q^m}^t $ and the map $ \phi : \mathcal{V} \longrightarrow \prod_{i= \ell+1}^{\ell + u} \mathbb{F}_q^{m_i \times n_i} $ given by
$$ \phi (\boldsymbol\lambda) = \pi \left( M_{\boldsymbol\gamma}^t(\boldsymbol\lambda) \right), $$
for $ \boldsymbol\lambda \in \mathcal{V} $, where $ \pi $ is as in Definition \ref{def pi}.
\end{definition}

The following result is straightforward using (\ref{eq inequality ranks for phi}).

\begin{proposition}
The map $ \phi : \mathcal{V} \longrightarrow \prod_{i= \ell+1}^{\ell + u} \mathbb{F}_q^{m_i \times n_i} $ from Definition \ref{def phi} is a vector space isomorphism such that $ {\rm wt}(\phi(\boldsymbol\lambda)) \geq {\rm wt}(\boldsymbol\lambda) $, for all $ \boldsymbol\lambda \in \mathbb{F}_{q^m}^t $. 
\end{proposition}

\subsection{Concrete examples of MSRD codes} \label{subsec examples 4 systematic}

We now provide examples of systematic matrices as in (\ref{eq systematic gen matrix}), and therefore examples of MSRD codes coming from Construction \ref{cons msrd systematic}. We will make use of the $ \mathbb{F}_{q^m} $-linear MSRD codes from Subsection \ref{subsec preliminaries msrd}.

Consider positive integers $ m = m_1 = \ldots = m_\ell $ and $ r = n_1 = \ldots = n_\ell = t \leq m $. Assume also that $ \ell + 1 = \mu (q-1) $ and let $ n = n_1 + \cdots + n_\ell = \ell r $, for some positive integer $ \mu $. Let $ a_1, \ldots, a_{q-1}, \beta_1, \ldots, \beta_{\mu r} \in \mathbb{F}_{q^m}^* $ satisfy the properties stated after equation (\ref{eq general msrd generator}). Set $ k = n+t-d+1 $ for some $ d \in [t+1,t+n] $. We may choose $ \mathcal{D}_0 \subseteq \mathbb{F}_{q^m}^{n+t} $ in Construction \ref{cons msrd systematic} as the $ \mathbb{F}_{q^m} $-linear MSRD code with generator matrix $ M_{t+k}(\mathbf{a},\boldsymbol\beta) \in \mathbb{F}_{q^m}^{(t+k) \times (n+t)} $, given in Subsection \ref{subsec preliminaries msrd}, or the $ \mathbb{F}_{q^m} $-linear MSRD code with parity-check matrix $ M_{n-k}(\mathbf{a},\boldsymbol\beta) \in \mathbb{F}_{q^m}^{(t+k) \times (n+t)} $, for the sum-rank length partition $ (n_1, \ldots, n_\ell, t) = (r, \ldots, r) $ ($ \ell+1 $ times). Observe that $ {\rm d}(\mathcal{D}_0) = d-t \geq 1 $ and $ \dim_{\mathbb{F}_{q^m}}(\mathcal{D}_0) = t+k $. Finally, by Gaussian elimination, we may obtain a generator matrix of $ \mathcal{D}_0 $ as in (\ref{eq systematic gen matrix}), for some $ \mathbf{g}_1, \ldots, \mathbf{g}_{t+k} \in \mathbb{F}_{q^m}^n $. 

The next step is to choose a matrix partition in order to define the vector space isomorphism $ \phi $ as in Subsection \ref{subsec examples 4 phi}. Let $ u $ be a positive integer and choose $ X_1, \ldots , X_u \subseteq [m] $ and $ Y_1, \ldots, Y_u \subseteq [t] $ such that $ (X_i \times Y_i) \cap (X_j \times Y_j) = \varnothing $ if $ i \neq j $. Define the $ \mathbb{F}_q $-linear subspace $ \mathcal{V} \subseteq \mathbb{F}_{q^m}^t $ and the vector space isomorphism $ \phi : \mathcal{V} \longrightarrow \prod_{i= \ell+1}^{\ell + u} \mathbb{F}_q^{m_i \times n_i} $ as in Definition \ref{def phi}.

By Construction \ref{cons msrd systematic}, we obtain an $ \mathbb{F}_q $-linear MSRD code $ \mathcal{C} \subseteq \prod_{i=1}^{\ell + u} \mathbb{F}_q^{m_i \times n_i} $ of minimum sum-rank distance $ {\rm d}(\mathcal{C}) = d \in [t+1,t+n] $ and dimension $ \dim_{\mathbb{F}_q}(\mathcal{C}) = m (n-d+1) + \sum_{i=\ell+1}^{\ell + u} m_in_i $, where 
$$ \ell = \mu (q-1)-1, \quad r = n_1 = \ldots = n_\ell \leq m = m_1 = \ldots = m_\ell, \quad m_{\ell+j} = |X_j| \quad \textrm{and} \quad n_{\ell + j} = |Y_j|, $$
for $ j \in [u] $. The possible values of $ \mu $ and $ r $ in this construction (which come from the code $ \mathcal{D}_0 $ from \cite{generalMSRD}) are described in \cite[Table 1]{generalMSRD}.

As a concrete example, we may choose $ \mu = 1 $ and $ r = m $, corresponding to linearized Reed--Solomon codes \cite{linearizedRS} (first row in \cite[Table 1]{generalMSRD}). In this case, we obtain an $ \mathbb{F}_q $-linear MSRD code in $ \prod_{i=1}^{\ell + u} \mathbb{F}_q^{m_i \times n_i} $, as above, of minimum sum-rank distance $ d \in [t+1,t+n] $, where
$$ \ell = q-2, \quad r = n_1 = \ldots = n_\ell = m_1 = \ldots = m_\ell, \quad m_{\ell+j} = |X_j| \quad \textrm{and} \quad n_{\ell + j} = |Y_j|, $$
for $ j \in [u] $.

\begin{remark} \label{remark const 4 better than 3}
By \cite[Th. 1]{doubly}, the vectors $ \mathbf{g}_1, \ldots, \mathbf{g}_{t+k} \in \mathbb{F}_{q^m}^n $ from the systematic generator matrix in (\ref{eq systematic gen matrix}) are such that the $ \mathbb{F}_{q^m} $-linear codes $ \mathcal{D}_I = \langle \mathbf{g}_i : i \in I \rangle_{\mathbb{F}_{q^m}} \oplus \langle \mathbf{g}_{t+1}, \ldots, \mathbf{g}_{t+k} \rangle_{\mathbb{F}_{q^m}} \subseteq \mathbb{F}_{q^m}^n $, for $ I \subseteq [t] $, are all MSRD with $ \dim_{\mathbb{F}_{q^m}}(\mathcal{D}_I) = k + |I| $. Thus we would be in the scenario of Subsection \ref{subsec examples 3}. However, using Construction \ref{cons msrd} in this case, we may extend such codes by adding any matrix sizes $ m_{\ell + 1} \times n_{\ell + 1} , \ldots, m_{\ell + u} \times n_{\ell + u} $, where
$$ m_{\ell + \ell_{i-1} + 1} n_{\ell + \ell_{i-1} + 1} + \cdots + m_{\ell + \ell_i} n_{\ell + \ell_i} \leq m, $$
for $ i \in [t] $, for integers $ 0 = \ell_0 < \ell_1 < \ell_2 < \ldots < \ell_t = u $. In particular, $ m_{\ell+1}n_{\ell +1 } + \cdots + m_{\ell + u}n_{\ell + u} \leq tm $.

However, the reader may easily verify that, using Construction \ref{cons msrd systematic}, we have more flexibility in the choice of the matrix sizes $ m_{\ell + 1} \times n_{\ell + 1} , \ldots, m_{\ell + u} \times n_{\ell + u} $ to extend the MSRD codes $ \mathcal{D}_I $. For instance, it is still necessary that $ m_{\ell+1}n_{\ell +1 } + \cdots + m_{\ell + u}n_{\ell + u} \leq tm $, but we can easily partition matrices in order to obtain $ m_{\ell + \ell_{i-1} + 1} n_{\ell + \ell_{i-1} + 1} + \cdots + m_{\ell + \ell_i} n_{\ell + \ell_i} > m $ for some $ i \in [t] $, which is not possible with Construction \ref{cons msrd}.

This is due to the fact that we are using a stronger property than \cite[Th. 1]{doubly}, namely, we are using that $ \mathcal{D}_0 $ is MSRD for the sum-rank length partition $ (n_1, \ldots, n_\ell, t) $ for $ t > 1 $.
\end{remark}

\begin{remark} \label{remark const 3 better than 4}
Conversely, it is natural to ask whether we may use Construction \ref{cons msrd systematic} for the doubly and triply extended MSRD codes that we could obtain via \cite[Th. 1]{doubly} from the lattices of MSRD codes in Subsection \ref{subsec examples 3}. However, such doubly and triply MSRD codes using \cite[Th. 1]{doubly} are extended by adding a Hamming-metric block (and extensions by adding a rank-metric block are not possible \cite[Prop. 11]{doubly}). Thus Construction \ref{cons msrd systematic} would not be applicable in this case.
\end{remark}

The previous two remarks show that, due to the concrete examples from Subsections \ref{subsec examples 3} and \ref{subsec examples 4 systematic}, one cannot always use Construction \ref{cons msrd systematic} instead of Construction \ref{cons msrd} and viceversa.

%
%

\section{Comparisons with previous MSRD codes} \label{sec comparisons}

In this section, we briefly compare the concrete examples of MSRD codes that can be obtained via Constructions \ref{cons product}, \ref{cons bases}, \ref{cons msrd} and \ref{cons msrd systematic} with the known MSRD codes in the literature \cite{alberto-fundamental, chen, linearizedRS, generalMSRD, doubly, twisted, neri-oneweight, santonastaso}. For simplicity, we will simply show that the parameters of the MSRD codes in those works can be obtained via Constructions \ref{cons product}, \ref{cons bases}, \ref{cons msrd} and \ref{cons msrd systematic}, whereas our constructions give rise to MSRD codes for strictly larger sets of parameters.

First, as stated at the end of Section \ref{sec cart products}, Construction \ref{cons product} does not cover new parameters, but can be decoded faster than linearized Reed--Solomon codes for the same parameters. 

Second, \cite[Const. VII.3]{alberto-fundamental} can be obtained applying Construction \ref{cons bases} recursively by choosing $ \ell = t = 1 $.

Next, the MSRD codes from \cite{twisted, santonastaso} cover the same parameters as linearized Reed--Solomon codes \cite{linearizedRS}, which in turn are a particular case of the MSRD codes from \cite{generalMSRD}. Now, the codes from \cite{generalMSRD} correspond to those in Subsection \ref{subsec examples 4 systematic} when choosing the trivial matrix partition $ X_1 = [m] $, $ Y_1 = [t] $ and $ u = 1 $ in order to construct the map $ \phi $ from Subsection \ref{subsec examples 4 phi}. Thus it is clear that the concrete MSRD codes from Subsection \ref{subsec examples 4 systematic} (built via Construction \ref{cons msrd systematic}) cover a strictly larger set of parameters.

Doubly extended linearized Reed--Solomon codes \cite{neri-oneweight} are a particular case of the doubly and triply extended MSRD codes from \cite{doubly}. Now, the doubly extended MSRD codes from \cite{doubly} correspond to those in Subsection \ref{subsec examples 3} when choosing $ \ell_1 = 1 $, $ \ell_2 = 2 $, $ m_{\ell+1} = m_{\ell + 2} = m $ and $ n_{\ell + 1} = n_{\ell + 2} = 1 $. Similarly, the triply extended MSRD codes from \cite{doubly} correspond to those in Subsection \ref{subsec examples 3} when choosing $ \ell_1 = 1 $, $ \ell_2 = 2 $, $ \ell_3 = 3 $, $ m_{\ell+1} = m_{\ell + 2} = m_{\ell + 3} = m $ and $ n_{\ell + 1} = n_{\ell + 2} = n_{\ell + 3} = 1 $. Hence it is clear that the concrete MSRD codes from Subsection \ref{subsec examples 3} (built via Construction \ref{cons msrd}) cover a strictly larger set of parameters.

The recent MSRD codes from \cite[Subsec. 5.2]{chen} can be obtained via Construction \ref{cons bases}, where the code $ \mathcal{C}_2 $ is the concrete MSRD code from Subsection \ref{subsec examples 3} choosing $ a_1 = 1 $ and puncturing the blocks corresponding to $ a_2, \ldots, a_{q-1} $ (i.e., choosing the generator matrix of a Gabidulin code \cite{gabidulin}), and restricting added blocks to square matrices, i.e., $ m_{\ell+1} = n_{\ell + 1}, \ldots, m_{\ell + \ell_2} = n_{\ell + \ell_2} $. Notice that the code $ \mathcal{C}_1 $ in Construction \ref{cons bases} needs to be a trivial code of dimension $ m_\ell $ by Theorem \ref{th bases}.

Finally, notice that Construction \ref{cons msrd systematic} cannot be obtained via Construction \ref{cons msrd} by Remark \ref{remark const 4 better than 3}. Similarly, Construction \ref{cons msrd} cannot be obtained via Construction \ref{cons msrd systematic} by Remark \ref{remark const 3 better than 4}. In particular, the concrete MSRD codes in Subsections \ref{subsec examples 3} and \ref{subsec examples 4 systematic} cover different sets of parameters.

\section*{Acknowledgement}

The author gratefully acknowledges the support from a Mar{\'i}a Zambrano contract by the University of Valladolid, Spain (Contract no. E-47-2022-0001486), and the support from MCIN/AEI/ \\10.13039/501100011033 and the European Union NextGenerationEU/PRTR (Grant no. TED2021-130358B-I00).

 

\end{document}